\documentclass[12pt]{article}

\usepackage[margin=1.2in]{geometry}
\usepackage{amssymb}
\usepackage{amsmath}
\usepackage{amsfonts}
\usepackage{amsthm}
\usepackage{amscd}
\usepackage{bm}

\usepackage[colorlinks=true]{hyperref}
\hypersetup{allcolors=[rgb]{0.1,0.1,0.4}}

\usepackage{tikz}
\usepackage{tikz-cd}

\newcommand{\B}{{\mathcal{B}}}

\newcommand{\bM}{\mathbb{M}}
\newcommand{\bS}{\mathbb{S}}
\newcommand{\bL}{\mathbb{L}}
\newcommand{\sS}{\mathcal{S}}

\newcommand{\bD}{\rm\bf D}

\newcommand{\Power}{2^}

\newcommand{\J}{\mathcal{J}}

\newcommand{\N}{{\mathcal{N}}}

\newcommand{\C}{{\mathcal{C}}}
\newcommand{\U}{\mathbb{U}}
\newcommand{\bU}{\mathbb{U}}
\newcommand{\LL}{{\mathcal{L}}}

\newcommand{\system}{{\rm\bf System}}
\newcommand{\phenome}{{\rm\bf Phenome}}

\newcommand{\<}{\langle}
\renewcommand{\>}{\rangle}

\newcommand{\sub}{{\rm\bf Sub}}

\newcommand{\ra}{\rightarrow}

\DeclareMathOperator{\fix}{fix}

\newcommand{\merge}{\wedge}
\DeclareMathOperator{\free}{Free}

\newtheorem{Thm}{Theorem}[subsection]
\newtheorem{Def}[Thm]{Definition}
\newtheorem{Pro}[Thm]{Proposition}
\newtheorem{Cor}[Thm]{Corollary}

\renewcommand{\P}{{\mathcal{P}}}
\renewcommand{\L}{{\mathcal{L}}}
\DeclareMathOperator{\colim}{colim}
\DeclareMathOperator{\agg}{agg}
\DeclareMathOperator{\id}{id}

\newcommand{\I}{{\mathcal{I}}}
\newcommand{\meet}{\wedge}
\newcommand{\join}{\vee}
\DeclareMathOperator{\eval}{eval}
\newcommand{\bZ}{\mathbb{Z}}

\title{On the mathematical structure of\\ cascade effects and emergent phenomena}
\author{Elie M. Adam\thanks{Both authors are with the Laboratory for Information and Decision Systems (LIDS) at the Massachusetts Institute of Technology (MIT).  Emails: \texttt{eadam@mit.edu} and \texttt{dahleh@mit.edu}} \and Munther A. Dahleh}
\date{}

\begin{document}

\maketitle

\begin{abstract}
  We argue that the mathematical structure, enabling certain cascading and emergent phenomena to intuitively emerge, coincides with Galois connections.  We introduce the notion of \emph{generative effects} to formally capture such phenomena.  We establish that these effects arise, via a notion of a \emph{veil}, from either concealing mechanisms in a system or forgetting characteristics from it.  The goal of the work is to initiate a mathematical base that enables us to further study such phenomena.  In particular, generative effects can be further linked to a certain loss of exactness.  Homological algebra, and related algebraic methods, may then be used to characterize the effects.
\end{abstract}

\section{Introduction.}
The principal thrust of the work aims at understanding and uncovering interaction-related effects that emerge from the interaction of several systems.  The typical situations of interest consist of a multitude of systems, coming together and interacting to produce a behavior that would not have emerged without interaction.  Such situations are fundamental and appear in countless settings exhibiting cascading phenomena and emergent behavior.

Our work thus begins with studying cascade effects. An intuitive instance of such effects is embodied by a trail of falling dominoes.  The fall of a domino triggers the fall of its successor in the sequence.  If the first domino falls, then the whole sequence of dominoes collapses by induction. On a more \emph{pressing} end, these effects appear in extinctions in ecosystems, spreads of epidemics, financial crises, power blackouts, cascading failures in infrastructures, propagation of delays and social adoption trends. There is then quite an interest, and an ever-increasing need, to understand such effects.  Many (mathematical) models are proposed to handle such situations.  We are however far from a fundamental understanding of the phenomenon.  Most of the intuition floating around is still intangible.  The term \emph{cascade effects} is itself ill-defined, vague and prone to many interpretation.  Some intuition does however exist to give rise to the term.  What mathematics are then needed to capture this intuition?

 We introduce, in this paper, the notion of \emph{generative effects}, to mathematically capture the essential intuition of cascade effects.  Generative effects may be understood as a broad---though formal---interpretation of the informal cascade effects.  As such, generative effects may be seen to also enclose some situations of emergent behavior, that need not exhibit the \emph{succession} of events exhibited by cascades.

 Generative effects will be seen to arise from certain Galois connections, and are induced by closure and kernel operators.  The paper will introduce, exemplify and study situations of generative effects.  The goal, in this paper, is to initiate a mathematical basis to further study such situations.  This work fits within an extensive effort (carried out in \cite{ADAM:Dissertation}) to characterize these phenomena.  In a more general setting, such effects are made to coincide with a \emph{loss of exactness}.  In such cases, we can then compute---via homological algebraic methods---(co)homology objects from the systems that capture their potential to produce effects, and use those objects to characterize the emerged generative effects.  We briefly return to this point at the end of this introduction.  The general setting, and the path to characterizing the effects through homological algebra, will however not be pursued in this paper.  We refer the reader to \cite{ADAM:Dissertation} for more details on that direction.  This paper is based on Chapter 3 of \cite{ADAM:Dissertation}.

\subsection{The mathematical sketch.}

To capture the intuition for cascade effects we require at least two ingredients.   The first is a notion of interaction or interconnection of systems.  However, such a notion by itself cannot give rise to interaction-related effects.  The interconnection of two systems may only give an interconnected system, and nothing more.  The second ingredient then consists of equipping the theory of interconnection with a notion of effects.  Such effects emerge only once we conceal parts or features from a system, by declaring what is observable from the system.  We thus envision the notion of a \emph{veil}, that covers parts of the system and leaves the rest---termed, the \emph{phenome}---bare and observable.  Generative effects then emerge whenever the phenomes of the separate systems cannot explain the phenome of the combined system.  The two ingredients may be summarized in the diagram:
\begin{equation*}
  \begin{CD}
     \<P,+\> @<\Phi<< \<S,\oplus\>.
  \end{CD}
\end{equation*}
The space $\< S,\oplus \>$ represents a space of systems.  For simplicity, $S$ is a set and $\oplus$ is a binary operator on $S$. The sum $s_1 \oplus s_2$ denotes the interaction of the systems $s_1$ and $s_2$.  The map $\Phi$ represents the veil, sending a system to its observable part, its phenome, in a set $P$. The set $P$ ought to be thought of as a space of simplified systems, and thus gains a notion of interaction through the $+$ binary operator.  Of course, the elements of this diagram are not arbitrary, and the complete mathematical diagram is slightly more intricate.  We will expound it in the next two subsections.  Regardless of the exact details, generative effects will be said to emerge whenever:
\begin{equation*}
  \Phi(s_1 \oplus s_2) \neq \Phi(s_1) + \Phi(s_2)
\end{equation*}
Effects are sustained by the veil whenever the phenome of the combined system cannot be explained by the phenome of the separate systems.  {\bf The inequality we obtain is then fundamental.}   If the inequality is not present, our intuition for cascade effects vanishes.  If the inequality is present, the intuition emerges.

\subsection{A contagion phenomenon.}

To dilute the abstraction and fill in mathematical details, we consider a situation of contagion.  A system will consist of an undirected graph. Each node in the graph can be either infected (active or failed) or healthy (inactive or non-failed), and is assigned an integer $k$ as a threshold.  All nodes are initially healthy.  A node then becomes infected if at least $k$ of its neighbors are infected.  Once a node is infected, it remains infected forever.  In this case, the order of infections does not affect the final set of infected nodes.  The system we described works on arbitrary undirected graphs, but for exposition, we will consider only systems defined on two nodes, as concrete examples.  Specifically, let us consider the systems, $S_1$ and $S_2$:
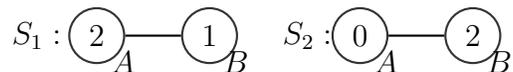
\begin{center}
\begin{tikzpicture}
  [scale=.5,auto=center,thick]
  \node at (-1.6,0) {$S_1:$};
  \node[circle,draw=black!80!white,scale=1] (n1) at (0,0) {$2$};
  \node at (0.7,-0.7) {$A$};
  \node[circle,draw=black!80!white,scale=1]  (n2) at (3,0)  {$1$};
  \node at (3.7,-0.7) {$B$};
  
  \node at (5.6,0) {$S_2:$};
  \node[circle,draw=black!80!white,scale=1]  (n11) at (7,0) {$0$};
    \node at (7+0.7,-0.7) {$A$};
    \node[circle,draw=black!80!white,scale=1]  (n22) at (10,0)  {$2$};
      \node at (7+3.7,-0.7) {$B$};
  
    \foreach \from/\to in {n1/n2,n11/n22}
    \draw (\from) -- (\to);
\end{tikzpicture}
\end{center}
The system $S_1$ can be summarized as ``{\tt if A is infected, then B becomes infected}'', while the system $S_2$ can be summarized as ``{\tt node $A$ is infected}''.  One can see that when the systems $S_1$ and $S_2$ are made to interact with each other, they will intuitively result in a cascade-like situation where $B$ becomes infected. The question is: how do we formalize such an intuition?  To formalize it we need two ingredients: a notion of interaction and a notion of \emph{effects} equipped on top of interaction.

Having two systems interact, or equivalently interconnecting two systems, consists of keeping the minimum of the thresholds.  The systems $S_1$ and $S_2$ interact to yield the system $S_1 \vee S_2$ given by:
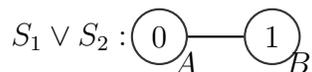
\begin{center}
\begin{tikzpicture}
  [scale=.5,auto=center,thick]
  \node at (-2.4,0) {$S_1\vee S_2:$};
  \node[circle,draw=black!80!white,scale=1] (n1) at (0,0) {$0$};
  \node[circle,draw=black!80!white,scale=1]  (n2) at (3,0)  {$1$};
    \node at (0.7,-0.7) {$A$};
  \node at (3.7,-0.7) {$B$};
    \foreach \from/\to in {n1/n2}
    \draw (\from) -- (\to);
\end{tikzpicture}
\end{center}
Interaction can be understood as combining the description of the two systems. The system $S_1 \vee S_2$ is in fact summarized as:
\begin{center}
  ``{\tt if A is infected, then B becomes infected}\\
  {\tt AND}\\ {\tt node A is infected}''.
\end{center}
The systems merge their update rules. The operator $\vee$ then gives us a notion of interaction through the merging of rules.  This notion by itself however does not account for the cascading phenomenon.

To retrieve the intuition, we need to focus on a particular feature of our systems.  Let $\Phi(S)$ denote the final set of infected nodes, that arise from system $S$.  By focusing on the final set of infected nodes, we have discarded any dynamics in the system that could potentially lead to more infections.  A set of infected nodes can be interpreted as a simplified system where thresholds are either $0$ or $\infty$ depending on whether they are infected ($0$) or not ($\infty$).  They thus inherit a notion of interaction, of merging descriptions, which coincides with set union.  Cascade-like intuition then arises because:
\begin{equation*}
  \Phi(S_1 \vee S_2) \neq \Phi(S_1) \cup \Phi(S_2)
\end{equation*}
Indeed, $\Phi(S_1) \cup \Phi(S_2) = \{A\}$ as $\Phi(S_1)=\{\}$ and $\Phi(S_2)=\{A\}$.  However, $\Phi(S_1 \vee S_2) = \{A,B\}$.  The observable part of the separate systems (i.e., their final set of infected nodes) fails to explain the observable part of the combined system.  The discarded mechanisms interact in the full systems to produce new observable that cannot be accounted for.

More generally, let $\Sigma$ denotes the set of nodes in the graph. If $2^\Sigma$ denotes the set of subsets of $\Sigma$, then the example lends itself to the following picture:
\begin{equation*}
\begin{CD}
  \big\< \Power \Sigma,\cup\big\> @<\Phi<< \big\< \system ,\vee\big\>
\end{CD}
\end{equation*}
The effects are now encoded in the \emph{inexactness} of the map $\Phi$. The map $\Phi$ is not unstructured, and possesses certain properties.  The paper sets forth a thesis that cascade-like effects arise from situations akin to the diagrammatic representation above.  The spaces of systems and phenomes vary, and are preordered sets. The map $\Phi$ tends to admit an adjoint, thus forming a Galois connection between systems and phenomes.  Elements of the spaces interact to yield their join (i.e., their least upper-bound) whenever it exists.  The effects are then sustained whenever $\Phi$ fails to commute with joins, i.e., the interaction operator.

\subsection{Summary.}

Generative effects are seen to arise from the pattern:
\begin{equation*}
\begin{CD}
  \big\< \phenome,\vee\big\> @<\Phi<< \big\< \system ,\vee\big\>.
\end{CD}
\end{equation*}
The space $\big\< \system ,\vee\big\>$ (resp. $\big\< \phenome,\vee\big\>$) is a preordered set of systems (resp. of phenomes) where every pair of systems (resp. of phenomes) interact to yield their least upper-bound, via $\vee$, the join operator. The map $\Phi$ acts as the described veil, partially concealing the systems and leaving the phenome bare. It admits an adjoint, thus forming a Galois connection. Its adjoint recovers from a phenome the simplest system explaining it.  Generative effects are said to be sustained whenever:
\begin{equation*}
  \Phi(S \vee S') \neq \Phi(S) \vee \Phi(S').
\end{equation*}
They are sustained whenever the phenome of the combined system cannot be explained by the phenomes of the separate systems.

\subsection{Outline of the paper}
The paper expounds this formalization.  It justifies the choices made, and equips the mathematics with the needed intuition.

We continue to elucidate the example on contagion in Section 3. We introduce the notion of  a veil and generative effects in Section 4.  Those can be seen to emerge from either concealing mechanisms in the systems (developed in Section 5) or forgetting characteristics (developed in Section 6).  Indeed, we establish in Section 7 that every veil can be factored into an instance of these two cases, and discuss its relation to Galois connections. We introduce, in Section 8, the notion of a dynamical veil to capture temporal aspects in cascade effects.  We finally develop techniques of factorization and lifts, in Section 9, to retrieve veils from \emph{non-veils}, and end with some remarks in Section 10.

\subsection{The goal of this line of research}
A fuller, more englobing, development of the concepts can be performed via the use of categories.  We however restrict to preordered sets to not introduce unnecessary complications for the readers.  Preorders can be trivially regarded as categories, and are thus a special case of the general concept.  Nevertheless, most of the analysis provided in this paper can be extended out to the more general case. For more details on the general case, we refer the reader to \cite{ADAM:Dissertation} Ch 8.  In the full generality of the theory, generative effects are linked to a certain \emph{loss of exactness}. Homological algebra then comes into the setting to analyze the situation. Specifically, we can extract algebraic object from the systems that encode their potential to generate effects. Those objects can then be used to understand the phenomenon, and link the behavior of the interconnected system to its separate constituents.  Mathematically, this picture enables us to develop (co)homology theories to understand cascade effects. The $0$th order (co)homology objects encode the phenome of the system, and higher order objects encode potentials to produce effects.  We refer the reader to \cite{ADAM:Dissertation} for a thorough development of this line of research.

\tableofcontents

\section{Mathematical preliminaries and definitions.}\label{sec:prelim}

A preordered set or \emph{proset} $\<S,\leq\>$ is a set $S$ equipped with a (binary) relation $\leq$ that is reflexive and transitive. If $\leq$ is also antisymmetric, then $\leq$ becomes a partial order and $\<S,\leq\>$ becomes a partially ordered set or \emph{poset}. A proset is said to be a join-semilattice (resp. meet-semilattice) if every pair of elements admits a least upper-bound, termed \emph{join} (resp. greatest lower-bound, termed \emph{meet}).  A proset that is both a join-semilattice and a meet-semilattice is said to be a lattice.  Note that if a proset $\<S,\leq\>$ admits finite joins (resp. finite meets) then $\leq$ is antisymmetric (see e.g., Proposition \ref{pro:antisymmetric} for more details).

A proset is said to be finitely cocomplete (resp. finitely complete) if every finite subset of it admits a join (resp. meet). A finitely cocomplete (resp. finitely complete) proset is then only a join-semilattice (resp. meet-semilattice) with a minimum (resp. maximum) element.
A proset $S$ is said to be cocomplete (resp. complete) if every subset of it admits a join (resp. meet).  A cocomplete proset is then necessarily complete: the lower-bounds of a subset admit a join by cocompleteness.  The converse also holds.   A complete lattice is then a lattice that admits arbitrary meets and joins.  A complete lattice is thus equivalently a cocomplete (resp. complete) preordered set.

\subsection{Notation.}  If $S$ is a set, then $2^S$ denotes the set of subsets of $S$.  If $S$ and $T$ are sets (resp.\ preordered sets), then $S^T$ denotes the set of maps (resp.\ order-preserving maps) from $T$ to $S$.  If $S$ and $T$ are prosets, then the set $S^T$ inherits a natural preorder relation $f \leq g$ if, and only if, $f(t) \leq g(t)$ for all $t \in T$. 

\section{The contagion phenomenon, revisited.} \label{sec:contagionSection}

The example presented in the introduction is only an instance of a more general class of systems.  A system in concern consists of $n$ nodes, or parts.  Each node can be either infected (active or failed) or healthy (inactive or non-failed), and is attributed a collection of neighborhood sets.  A neighborhood set is only a subset (possibly empty) of the $n$ nodes. Each node can be attributed either one, multiple or no neighborhood sets.  A node becomes infected if all the nodes in (at least) one of its neighborhood sets are infected.  Once a node becomes infected, it remains infected forever.  Again, the order of infections does not affect the set of final infected nodes.

The example presented in the introduction can be seen as a special case where the neighborhood sets of node $i$ are only subsets of cardinality $k$, the threshold of $i$. The operator $\Phi$ described, and its \emph{inexactness}, carries through unchanged.

\subsection{Syntax and interpretation.}

Let $\Sigma:=\{a,b,\cdots,h\}$ be a finite set of $n$ elements.  The set $\Power\Sigma$ denotes the set of subsets (or powerset) of $\Sigma$.  For notational convenience, we define $\bD$ (for Description) to be the set:
\begin{equation*}
  \Sigma \ra \Power{\Power\Sigma}
\end{equation*}
A map in $\bD$ assigns to every element of $\Sigma$ a collection of subsets of $\Sigma$.  A system, as presented at the start of the section, is syntactically described by a map $\N \in \bD$.  Conversely every map in $\bD$ is a meaningful syntactic description of a system.

\subsubsection{Interpretation.} The syntactical description $\N$ is interpreted as follows.  A system is made up of $n$ part, labeled say $a,b,\cdots,h$. To each part $i$ is assigned a collection of neighborhood sets $\N(i)$. Every part can be either infected (active or failed) or healthy (inactive or non-failed). All parts are initially inactive. Part $i$ is infected at time $m+1$ if, and only if, either it was infected at time $m$ or all the parts in some neighborhood set $N \in \N(i)$ of $i$ are infected at time $m$.  Thus, once a part is infected, it remains infected forever.

Let $A_m$ denote the set of infected parts at time $m$.  We initiate $A_0$ to be the empty set, and recursively define $A_1, A_2, \cdots$ such that $i \in A_{m+1}$ if, and only if, either $i \in A_m$ or $N \subseteq A_m$ for some set $N \in \N(i)$. Therefore, every map in $\bD$ assigns to part $i$ a monotone (or order-preserving) Boolean function $\phi_i : \Power\Sigma \ra \Power{\{*\}}$.  The set $\{*\}$ denotes the set with one element. Then:
\begin{equation*}
  i \in A_{m+1}   \quad \text{ iff } \quad \text{ either } i \in A_m \text{ or } \phi_i(A_m) = \{*\}.
\end{equation*}
Whenever the set $\Sigma$ is finite, the dynamics converge after finitely many steps.

\begin{Pro}\label{Pro:Conv}
 If $\Sigma$ has cardinality $n$, then $A_n = A_{n+1}$.
\end{Pro}

\begin{proof}
  If $A_n \neq A_{n+1}$, then $A_m \neq A_{m+1}$ for $0 \leq m \leq n-1$. Thus, if $A_n \neq A_{n+1}$ then $A_n$ would contain more than $n$ elements.
\end{proof}
We thus refer to the final set of infected nodes as $A_\infty$.  The set $A_\infty$ is only used to correspond to the case where $A_0$ is initialized to the empty set.

\subsubsection{Interaction.} Two systems are made to syntactically interact by merging their descriptions.  Syntactic descriptions $\N \in \bD$ and $\N'\in \bD$ interact by yielding their union $\N \cup \N' \in \bD$ where $(\N\cup \N')(i) = \N(i) \cup \N'(i)$.  The collection of neighborhoods are combined. Indeed, we can order $\bD$ by inclusion as $\N \subseteq \N'$ if $\N(i) \subseteq \N'(i)$ for all $i$.  Every pair of descriptions $\N$ and $\N'$ in the partially ordered set $\bD$ admits a least upper-bound denoted by $\N \cup \N'$.

\subsection{Semantics.}

To study the systems, we will recover from a syntactical description $\N \in \bD$, a map $f_\N: \Power\Sigma \ra \Power\Sigma$ that sends $S \subseteq \Sigma$ to the final set of infected nodes if $A_0$ was initialized to $S$, i.e., if all the parts in $S$ are initially infected.

\begin{Pro}
  The set $A_\infty$ corresponding to the final set of infected nodes is $f_\N(\emptyset)$.
\end{Pro}
\begin{proof}
  The proof is immediate by definition of $A_\infty$.
\end{proof}
As mentioned, the set $A_\infty$ is only used to correspond to the case where $A_0$ is initialized to the empty set.

The map $f_\N$ derived from a syntactical description $\N$ may be thought of as being the object that gives a meaning to the description, the semantics behind the syntax. Different syntactical descriptions in $\bD$ may yield the same system object.  However, different system objects refer to different systems, when it comes to looking at the final set of infected nodes.  The system object $f_\N$ can be seen as the representation-independent object we are interested in.  We will thus refer, in this section, to $f_\N$ as the system (object), as opposed to $\N$ which is referred to as the system syntax or description.

\begin{Def}
  A map $f : \Power\Sigma \ra \Power\Sigma$ is said to satisfy A.1, A.2 and A.3, respectively, if:
    \begin{itemize}\setlength\itemsep{0em}
    \item[A.1.] $S \subseteq f(S)$ for all $S \subseteq \Sigma$.
    \item[A.2.] If $S\subseteq S'$ then $f(S) \subseteq f(S')$, for all $S, S' \subseteq \Sigma$.
    \item[A.3.] $f f(S) = f(S)$ for all $S \subseteq \Sigma$.
  \end{itemize}
\end{Def}

\begin{Pro}
  If $\N \in \bD$, then $f_\N: \Power\Sigma \ra \Power\Sigma$ satisfies A.1, A.2 and A.3.
\end{Pro}

\begin{proof}
  The axioms A.1, A.2 and A.3 immediately follow from the description of the systems.  We refer the reader to \cite{ADAM:AlgebraCascadeEffects} for more details.
\end{proof}
Conversely, we have:

\begin{Pro}
 If $f: \Power\Sigma \ra \Power\Sigma$ satisfies A.1, A.2 and A.3, then $f = f_\N$ for some description $\N \in \bD$.
\end{Pro}

\begin{proof}
 Construct $\N$ such that $\N(i) = \big\{ S \in \Power\Sigma : i \in f(S)\big\}$.
\end{proof}
We then define $\system$ to be the set of maps satisfying A.1, A.2 and A.3.  The maps in $\system$ are often known as \emph{closure operators}. On one end, they appeared in the work of Tarski (see e.g., \cite{TAR1936} and \cite{TAR1956}) to formalize the notion of deduction. On another end, they appeared in the work of Birkhoff, Ore and Ward (see e.g., \cite{BIR1936}, \cite{ORE1943} and \cite{WAR1942}, respectively), parts of foundational work in universal algebra.  The first origin reflects the consequential relation in the effects considered.  The second origin reflects the theory of interaction of multiple systems.  Closure operators appear as early as \cite{MOO1911}. 

If we order $\system$ by:
\begin{equation*}
  f \leq g  \quad \text{ if } \quad f(S) \leq g(S) \text{ for all } S \subseteq \Sigma,
\end{equation*}
then $\<\system,\leq\>$ becomes a partially ordered set. The relation $f \leq g$ can be thought of as $f$ is a subsystem of $g$. Furthermore, every pair of systems $f$ and $g$ admits an upper-bound $f \vee g$.

\begin{Pro}
 If $\Sigma$ has cardinality $n$, then $f \vee g = (fg)^n$.
\end{Pro}

\begin{proof}
  The map $(fg)^n$ satisfies A.1, A.2 and A.3, and belongs to $\system$.  Indeed, A.1 and A.2 are preserved by composition. The axiom A.3 is satisfied as $(fg)^n(S) = (fg)^{n+1}(S)$ whenever $\Sigma$ has cardinality $n$, by an argument similar to that in the proof of Proposition \ref{Pro:Conv}. Finally, if $h \in \system$ and $f \vee g \leq h$, then $hf = fh = h$ and $hg = gh = h$. Thus $(fg)^n \leq (fg)^nh \leq h$ whenever $f \vee g \leq h$.
\end{proof}
The poset $\system$ is then a join-semilattice $\<\system, \leq, \vee\>$.  The semilattice also admits meets (i.e., greatest lower-bounds) making it a lattice. Its minimum element is the identity map, while its maximum element is the map $- \mapsto \Sigma$.

Two systems can be seen to interact (semantically) by iteratively applying their system maps till they yield an idempotent map, i.e., satisfying A.3. The properties A.1 and A.2 are always preserved under composition. Most importantly, the semantical interaction of systems coincides with the syntactical interaction of systems.

\begin{Pro}
 If $\N$ and $\N'$ are descriptions, then $f_{\N \cup \N'} = f_\N \vee f_{\N'}$.
\end{Pro}

\begin{proof}
  We have $f_{\N \cup \N'} (S) = S$ if, and only if, whenever $N \in \N\cup\N'(i)$ lies in $S$, then $i \in S$. Or equivalently if, and only if, $f_\N(S) = f_{\N'}(S) = S$.  Furthermore, the fixed-points of $f_\N \vee f_{\N'}$ are the sets that are fixed-points of both $f_\N$ and $f_{\N'}$.  The result then follows as the maps in $\system$ are uniquely determined by their fixed-points.  See e.g., Theorem \ref{Thm:Iso} or \cite{ADAM:AlgebraCascadeEffects} for more details on the last assertion. 
\end{proof}
We established thus far a theory of interconnection, via the space $\<\system, \vee\>$.  However, interconnecting two systems will only give us an interconnected system.  No cascading phenomena are yet present.  Those will only emerge once we decide to conceal features in the systems.

\subsection{The contagion intuition.}

To recover the intuition, we conceal the dynamics. We do so by only keeping the final set of infected nodes.  We are thus observing from our systems, subsets of $\Sigma$ corresponding to the final set of infected nodes. To this end, we define $\Phi : \system \ra \Power\Sigma$ to be the map sending $f$ to its least fixed-point. Such a map is well defined, as:
\begin{Pro}\label{pro:CompLattice}
  If $f:\Power\Sigma \ra \Power\Sigma$ satisfies A.1, A.2 and A.3, then its set of fixed-points $\fix(f) = \{S : fS=S\}$ when ordered by inclusion forms a complete lattice.  Furthermore, if $S$ and $S'$ are fixed-points of $f$, then $S \cap S'$ is a fixed-point of $f$.
\end{Pro}
\begin{proof}
  Let $S, S' \in \fix(f)$ be fixed-points. We have $f(S \cap S') \leq f(S) = S$ and $f(S \cap S') \leq f(S') = S'$ by A.2. As $S\cap S' \leq f(S \cap S')$ by A.1, we get $f(S \cap S') = S \cap S' \in \fix(f)$.  Let $\sS$ denote the collection of fixed-points that contain both $S$ and $S'$, namely:
  \begin{equation*}
    \sS := \{T \in \fix(f): S \subseteq T \text{ and } S'\subseteq T\}
  \end{equation*}
  As $\Sigma \in \fix(f)$, the set $\sS$ is non-empty.  The least-upper-bound of $S$ and $S'$ in $\fix(f)$ is then $\bigcap \sS$, the intersection of all the sets in $\sS$.  The greatest lower-bound of $S$ and $S'$ is $S \cap S'$.  The set $\fix(f)$ then forms a lattice.  The lattice $\fix(f)$ is complete as it is finite.
\end{proof}
The map $\Phi$ is then well defined, as a complete lattice admits a minimum element.  This minimum element corresponds to the set-intersection of all the fixed-points of $f$.

We term our observations, namely the subsets of $\Sigma$, as phenomes.  We also refer to $\Power\Sigma$ as the space of phenomes, denoted by $\phenome$.

\begin{Pro} The map $\Phi : \system \ra \phenome$ satisfies:
\begin{itemize}\setlength\itemsep{0em}
  \item[P.1.] If $f \leq g$ in $\system$, then $\Phi(f) \subseteq \Phi(g)$.
  \item[P.2.] If $S\in \phenome$, then $\{f : S \subseteq \Phi(f)\}$ has a (unique) minimum element.
  \end{itemize}
\end{Pro}

\begin{proof}
  (P.1) If $f \leq g$ then $\{S : fS = S\} \supseteq \{S: gS = S\}$. (P.2) For every $S$, the system $- \mapsto -\cup S$ is the minimum element of $\{f : S \subseteq \Phi(f)\}$. 
\end{proof}
First, the map $\Phi$ is order-preserving, and thus preserves the subsystems relation among the systems.  Second, every set of infected nodes can be lifted to a simplest system explaining that set.

  Such a map $\Phi$, satisfying P.1 and P.2, from the space of systems to a space of phenomes is termed a \emph{veil}.  In this context, contagion phenomena (later termed generative effects) arise precisely whenever:
  \begin{equation*}
    \Phi(f \vee g) \neq \Phi(f) \vee \Phi(g).
  \end{equation*}
  They arise whenever keeping only the final set of infected nodes cannot account for what happens when the two systems interact.  Indeed, the mechanisms that we have concealed interact and activate, or infect, more nodes than we can observably account for.  The phenomenon is now encoded in the \emph{inexactness} of the veil $\Phi$.  The inequality is the essential point.  

We return to this example as the paper unfolds.  We first set out to describe the general structure of the situation, and formally introduce generative effects.

\section{The two ingredients, formalized.}

As illustrated, two ingredients are required to sustain generative effects.  We first need a theory of interaction or interconnection of systems.  A theory of interconnection by itself cannot, however, account for such phenomena.  We need a notion of a veil, that conceals features from a system, and keeps a phenome observable.  Generative effects then emerge whenever the phenome of the combined system cannot be explained by the phenomes of the separate systems.

\subsection{Interaction of systems.}

 Let $\system$ be a preordered set, namely a set equipped with a binary relation $\leq$ that is reflexive and transitive.  Each element of $\system$ is considered to be a system. The $\leq$ relation dictates how the systems are related to each other. 

\begin{Def}
 A system $s$ is said to be a subsystem of $s'$ if $s \leq s'$.
\end{Def}
 Two systems will interact to yield their least upper-bound, only if it exists.  We will generally consider that least upper-bounds of finite subsets always exist, as such conditions will (or can be made to) be satisfied in most of our situations in concern. A preordered set is said to be finitely cocomplete, if every finite subset of it admits a (unique) least upper-bound.

\begin{Pro}\label{pro:antisymmetric}
  If $\system$ is finitely cocomplete, then the relation $\leq$ is antisymmetric.  
\end{Pro}

\begin{proof}
   If $s \leq s'$ (resp. $s' \leq s$) then $s'$ (resp. $s$) is the least upper-bound of $s$ and $s'$. As the least upper-bound is unique by definition, we get $s = s'$ whenever $s \leq s'$ and $s' \leq s$.
\end{proof}
If $\system$ is finitely cocomplete, then $\leq$ becomes a partial order.  {\bf In this paper, we consider $\system$ to be a finitely cocomplete preordered set, unless indicated otherwise.} %
A finitely cocomplete preordered set always admits, by definition, a minimum element: the least upper-bound of the empty set.

\begin{Def}
   Two systems $s$ and $s'$ interact to yield their least upper-bound $s \vee s'$. More generally, a finite subset of systems $S \subseteq \system$ interacts to yield its least upper-bound $\vee S$ as a resulting system.
\end{Def}
 A collection of systems can only interact in a unique way and it is via the $\vee$ operator, to yield their least upper-bound. The binary operator $\vee$ is associative, commutative and idempotent.  The algebra $\<\system,\leq,\vee\>$ is usually termed a join-semilattice.

\subsubsection{Remark.} Conversely, every associative, commutative and idempotent binary operator on a set induces a partial order on it. The development could have thus began with a join semilattice. However, the order relation is seen to be more essential than the join operation.  This is especially true in the general level of the developed theory, through the use of categories and functors.  In the general case, the order relation is replaced by sets of morphisms and joins are replaced by colimits.  This direction will however not be considered in this paper.

\subsubsection{Remark.}  The subsystem relation may admit various interpretations.  The notion of interconnection advocated by the behavioral approach to systems theory (see e.g., \cite{POL1998}, \cite{WIL2007}, and Subsection \ref{sec:behApp}) can be seen as a special case of that developed in this section.  Indeed, the notion of a subsystem in this section translates to a reverse inclusion of behaviors.  Furthermore, interpreting $s \leq s'$ as $s$ being a partial description or an approximation of $s'$ is reminiscent of ideas in \cite{SCO1972}, \cite{SCO1972B} and \cite{SCO1972A}.  The implication of such a connection will not however be investigated in this paper.  Such a direction of research may however well be fruitful.

\subsection{Interlude on capturing the generativity of a system.}

Let us informally consider a generative grammar (see e.g. \cite{CHO1965} and related work for a formal treatment) to be a collection of rules that dictates which sentences can be formed.  Every grammar then builds one language, and different grammars may describe the same language. The grammars generating a same language are, however, different: adding a rule to one grammar could yield a very different effect on the language than adding it to another grammar.  Similar effects occur in deduction (as seen through contagion in Section \ref{sec:contagionSection}) and in situations exhibiting cascade effects or emergent phenomena.  How do these grammars then gain this \emph{generativity}?  It is definitely coming from their grammar rules.  Yet, how do we capture it?  We capture it by \emph{destroying} the rules, and studying how the grammar in full and the grammar without the rules (amounting to only the language) behave when combined with other grammars.  It is the vivid discrepancy in interaction outcome between the presence of the rules and their absence that encodes the generativity.  To then capture cascade effects resulting from the interaction of systems, we perform the following experiment. On one end, we let the systems interact and observe the outcome of the interaction.  On another end, we destroy the potential a system has to produce effects let them interact without it.  These two ends, in the presence of cascade effects, will show a discrepancy in interaction outcome.  This discrepancy then encodes the phenomenon.  Studying the discrepancy amounts to studying the phenomenon.

\subsection{Veils and generative effects.}

Generative effects are seen to emerge when we decide to focus on a particular property of a system.  Such a focus is achieved by declaring a map $\Phi: \system \ra P$, termed a \emph{veil}, from the set of systems, to a set of observables, termed \emph{phenomes}. Phenomes can be properties, features or even subsystems of a particular system.  They ought to be thought of as simplified systems, and thus inherit an order-relation and a notion of interaction. The space $P$ is then, in turn, a cocomplete preordered set.

\begin{Def}
 A veil on $\system$ is a pair $(P,\Phi)$ where $\<P,\leq\>$ is a finitely cocomplete preordered set, and $\Phi : \system \ra P$ is a map such that:
\begin{itemize}\setlength\itemsep{0em}
  \item[V.1.] The map $\Phi$ is order-preserving, i.e., if $s \leq s'$, then $\Phi s \leq \Phi s'$.
  \item[V.2.] Every phenome admits a simplest system that explains it, i.e., the set $\{s: p \leq \Phi s\}$ has a (unique) minimum element for every $p \in P$.
\end{itemize}%
\end{Def}
As the map $\Phi$ always subsumes a codomain, we will often refer to $\Phi$ as the veil, instead of the pair $(P,\Phi)$.  However, viewing a veil as a pair $(P,\Phi)$ highlights an important point. We may define different veils for a same space of systems, and each veil would define the space of phenomes to be observed from the system.  The picture to keep in mind then is not that of fixing $\system$ and $P$ and varying a veil in between.  It is of fixing $\system$ and varying $(P,\Phi)$ to yield different facets of the systems.

The veil is intended to hide away parts of the system, and leave other parts, the phenome, of the system bare and observable.  The axiom V.1 indicates that veiling a subsystem of a system may only yield a subphenome of the phenome of the system.  The axiom V.2 indicates that everything one observes can be completed in a simplest way to something that extends under the veil. Generative effects occur precisely when one fails to explain the happenings through the observable part of the system.  In those settings, the things concealed under the veil would have interacted and produced observable phenomes. 

\begin{Def}
 A veil $(P,\Phi)$ is said to sustain generative effects if $\Phi(s \vee s') \neq \Phi (s) \vee \Phi(s')$ for some $s$ and $s'$.
\end{Def}
Different veils may be defined for the same space $\system$.  Some will sustain generative effects and some will not.  For instance, both veils $(\system,\id)$ and $(\{*\},*: \system \ra \{*\})$ do not sustain generative effects at all.  The veil $(\system,\id)$ , being the identity map, hides nothing, while the veil $(\{*\},*: \system \ra \{*\})$ hides everything.  All that can be observed is explained by what is already observed.  Thus the standard intuition for systems exhibiting cascading phenomena, or contagion effects, does not stem from a property of a system.  It is rather the case that the situation admits a highly suggestive phenome and a highly suggestive veil that sustains such effects.  Those effects are thus properties of the situation.  Should we change the veil, we may either increase those effects, diminish them or even make them completely go away. Such interaction-related effects depend only on what we wish to observe.

The first property of the veil is somewhat self-explanatory. It ensures that the map respects the relation among systems and is compatible with the preorders.  The second property, is less transparent, but gives the map a \emph{generative} intuition present in cascading phenomena.  To explain the second property, we note that generative effects can be seen to arise from two situations.  We either conceal mechanisms in the systems, or we forget characteristics of the systems.  These two situations will be expounded in the next two sections.

To make the space of phenomes $P$ explicit in the paper, as done with $\system$, we will often refer to $P$ as $\phenome$.  Such a reference is mainly done in the following two sections. {\bf Similarly to $\system$, we consider $\phenome$ in this paper to be a finitely cocomplete preordered set, unless indicated otherwise.}

\section{Concealing mechanisms.}
The first source of generative effects consists of concealing mechanisms, or dynamics, in a given system.  Two systems, sharing a same phenome, may become identical once the mechanisms are concealed.  Their potential to produce effects in the phenome, while interacting with other systems, may however be different. Indeed, concealed mechanisms may play a role upon the interaction of systems.

To conceal (or destroy) mechanisms in a system, we require a map $\kappa: \system \ra \system$ satisfying:
\begin{itemize} \setlength\itemsep{0em}
  \item[K.1.] $\kappa(s) \leq s$ for all $s$.
  \item[K.2.] If $s \leq s'$, then $\kappa(s) \leq \kappa(s')$, for all $s$ and $s'$.  
  \item[K.3.] $\kappa\kappa(s) = \kappa(s)$ for all $s$. 
\end{itemize}
First, the map $\kappa$ reduces a system to a subsystem of it.  Second, the map $\kappa$ preserves the relation among systems.  Third, the map $\kappa$ does not discard anything from a system whose mechanisms are already discarded.  An operator satisfying K.1, K.2 and K.3 are usually termed kernel operators. 

The operator $\kappa$ can be intuitively expected to sustain generative effects whenever $\kappa( s\vee s') \neq \kappa(s) \vee \kappa(s')$ for some $s$ and $s'$. In such a case, the mechanisms concealed interact and produce more than what can only be produced by the phenomes. Put differently, the discarded mechanisms have a role to play in the interaction of systems with respect to our simplistic view, as phenomes, of the systems.

Let us define $\phenome \subseteq \system$ to be the set of fixed-points $\{s: \kappa(s) = s\}$ of $\kappa$.  Then by K.3, we get $\kappa(\system) = \phenome$.  We may then define $\pi : \system \ra \phenome$ such that $\pi(s) = \kappa(s)$. Every kernel operator on $\system$ thus gives rise to a surjective veil:

\begin{Pro}
 The map $\pi$ is surjective and order-preserving, and for every $p \in \phenome$, the set $\{s : p \leq \pi(s) \}$ has a minimum element.
\end{Pro}

\begin{proof}
 As $\pi(p) = \kappa(p) = p$ for every $p \in \phenome$, the map $\pi$ is surjective. The map $\pi$ is clearly order-preserving.  Finally, the set $\{s : p \leq \pi(s) \}$ has $p$ itself as a minimal element.
\end{proof}
Conversely, every surjective veil induces a kernel operator on $\system$:

\begin{Pro}\label{Pro:KernelVeilConverse}
  If $\pi: \system \ra \phenome$ is a surjective order-preserving map such that $\{s : p \leq \pi(s) \}$ has a minimum element for every $p$, then there exists a unique injective order-preserving map $i: \phenome \ra \system$ such that $\pi i$ is the identity map on $\phenome$, and $i\pi$ is a kernel operator on $\system$.
\end{Pro}

\begin{proof}
 For every $p \in \phenome$ define $i(p)$ to be the minimum element of $\{s: p \leq \pi(s)\}$.  As $\pi$ is surjective, it follows that $\pi i$ is the identity. The map $i$ is then injective. The map $i\pi$ is a kernel operator as $i\pi(s) = \min\{s' : \pi(s) \leq \pi(s')\}$.  The requirements K.i can then be easily checked. Uniqueness of $i$ follows from Proposition \ref{Pro:Galois}(i).
\end{proof}
Finally, whether or not generative effects are sustained by the veil $\pi$, depends on the properties of the kernel operator $\kappa$.

\begin{Pro} If $s, s' \in \system$, then:
  \[\pi( s\vee_{\system} s') \neq \pi(s) \vee_{\phenome} \pi(s') \text{ iff } \kappa( s\vee_{\system} s') \neq \kappa(s) \vee_{\system} \kappa(s').\]
\end{Pro}

\begin{proof}
  If $s$ and $s'$ are fixed-points of $\kappa$, then their join in $\system$ coincides with their join in $\phenome$.  Indeed, if $\kappa(s)=s$ and $\kappa(s')=s'$, then $\kappa(s \vee s') = s \vee s'$ by K.1 and K.2.
\end{proof}
We next provide some example situations of concealing mechanisms.

\subsection{Contagion and deduction systems.}

We return to our contagion example.  Recall that a system corresponds to a map $f: \Power\Sigma \ra \Power\Sigma$ satisfying:
  \begin{itemize}
    \item[A.1.] $S \subseteq f(S)$ for all $S$.
    \item[A.2.] $f(S) \subseteq f(S')$ if $S\subseteq S'$ for all $S$ and $S'$.
    \item[A.3.] $f f(S) = f(S)$ for all $S$.
  \end{itemize}
  The poset $\system$ corresponds to the set of maps satisfying A.1, A.2 and A.3, ordered by $f \leq g$ if $f(S) \leq g(S)$ for all $S$.  Every pair $f, g \in \system$ admits a least upper-bound $f \vee g \in \system$.  The poset $\phenome$ corresponds to $\Power\Sigma$ ordered by inclusion.  Two sets in $\Power\Sigma$ admit their set-union as the least upper-bound. We then define $\pi: \system \ra \phenome$ that sends $f$ to its least fixed-point.
  \begin{Pro}
    The map $\pi$ is a surjective veil.
  \end{Pro}
  \begin{proof}
   The map $\pi$ is surjective as it maps every system $- \mapsto - \cup S$ into $S$. The set $\{s : S \leq \pi(s)\}$ also has $- \mapsto - \cup S$ as a minimum element.  It is also clearly order-preserving.
  \end{proof} 
If we define $i : \phenome \ra \system$ to be the map $S \mapsto - \cup S$, then $i\pi$ yields a kernel operator.  This kernel operator can be interpreted as destroying all the potential a system has to infect additional nodes when interacting with other systems.  The kernel operator yields the simplest system (with respect to the dynamics) that can account for the infected nodes at the end.

 \subsubsection{A dual perspective.}
  
Let $\LL$ denote the collection of subsets of $\Power\Sigma$ such that (i) $\Power\Sigma \in \LL$ and (ii) if $A, B \in \LL$ then $A \cap B \in \LL$. The sets in $\LL$ are sometimes called \emph{Moore families} or \emph{closure systems}.  The set $\LL$ may then be ordered by reverse inclusion, to yield a lattice whose join is set-intersection.

\begin{Thm}\label{Thm:Iso}
  The map $f \mapsto \{S : f(S) = S\}$ defines an isomorphism between $\<\system , \leq , \vee\>$ and $\<\LL, \supseteq, \cap\>$
\end{Thm}

\begin{proof}
 Such a fact is well known regarding closure operators.  We refer the reader to \cite{ADAM:AlgebraCascadeEffects} for the details.
\end{proof}
As such, every system can be uniquely identified with its set of fixed-points.  Interaction of systems then consists of intersecting the fixed-points. As a consequence:

\begin{Cor}
  $\pi(f)$ is the intersection of all the fixed-points of $f$. \qed
\end{Cor}
For more details on the properties of such systems, we refer the reader to \cite{ADAM:AlgebraCascadeEffects}.  This line of example first aimed at understanding the mathematical structure underlying models of diffusion of behavior commonly studied in the social sciences.  The setup there consists of a population of interacting agents.  In a societal setting, the agents may refer to individuals.  The interaction of the agents affect their behaviors or opinions.  The goal is to understand the spread of a certain behavior among agents given certain interaction patterns.  Threshold models of behaviors (captured by M.0, M.1, M.2 and M.3 in \cite{ADAM:AlgebraCascadeEffects}) have appeared in the work of Granovetter \cite{GRA1978}, and more recently in \cite{MOR2000}.  Such models are key models in the literature, and have been later considered by computer scientists, see. e.g., \cite{KLE2007} for an overview.

\subsection{Relations and projections.}

A relation $R$ between sets $A$ and $B$ is a subset of $A \times B$. The set of relations between $A$ and $B$, ordered appropriately, admits two canonical veils.

First, define $\pi : (\Power{A \times B},\subseteq) \ra (\Power A,\subseteq)$ to send $R$ to $\{a : (a,b) \in R \text{ for all }b\}$. The map $\pi$ is a surjective veil. Indeed, the set $\{ R : S \subseteq \pi R \}$ has $S \times B$ as a minimum element.  This veil sustains generative effects as generally:
\begin{equation*}
  \pi( R \cup R' ) \neq \pi(R) \cup \pi(R')
\end{equation*}

Second, define $\pi' : (\Power{A \times B},\supseteq) \ra (\Power A,\supseteq)$  that sends $R$ to $\{a : (a,b) \in R \text{ for some }b\}$. The map $\pi'$ is also a surjective veil. Indeed, the set $\{ R : S \supseteq \pi' R \}$ has $S \times B$ as a \emph{minimum} element with respect to $\supseteq$. This veil also sustains generative effects as generally:
\begin{equation*}
  \pi( R \cap R' ) \neq \pi(R) \cap \pi(R')
\end{equation*}
The posets $(\Power{A \times B},\subseteq)$ and $(\Power{A \times B},\supseteq)$ of systems are dual to each other. The join corresponds to set-union in the first, whereas it corresponds to set-intersection in the second.  The veil $\pi'$ can be further interpreted in systems-theoretic situations, through the next example.

\subsection{Subsystem behavior or concealing parameters.}\label{sec:behApp}

In the behavioral approach to systems theory, a system is viewed as a pair of sets $(\U,\B)$. The set $\U$---termed, the \emph{universum}---depicts the set of all possible outcomes or trajectories.  The set $\B$ is a subset of $\U$---termed the \emph{behavior}---that defines which outcomes are deemed allowable by the dynamics of the system.  The sets $\U$ and $\B$ can be equipped with various mathematical structures to suit various need.  We will however, without loss of generality, only be concerned with sets, without any additional structure. Interconnecting two systems $(\U,\B)$ and $(\U,\B')$ yields the system $(\U,\B\cap\B')$.  Indeed, the interconnected systems keeps the trajectories that are deemed possible by the separate systems.  We refer the reader to \cite{POL1998} and \cite{WIL2007} for more details on the behavioral approach.

In this subsection, we are interested in the behavior of a subsystem of a system $(\U,\B)$ as some changes are incurred into the greater system.  A change in this setup is depicted as another system $(\U,\C)$. Incurring the change then consists of obtaining the system $(\U,\B\cap\C)$. To define a subsystem, we project $\U$ onto a smaller universum.  For instance, let us suppose $\U = \bS \times \bS'$ is a product space.  Projecting $\B$ canonically onto the universum $\bS$ yields the behavior of the subsystem of $(\U,\B)$ living in the universum $\bS$.

We thus define $\system$ to be $\Power\U$, ordered by reverse inclusion.  Two behaviors $\B$ and $\B'$ then admit a least upper-bound (with respect to the reverse inclusion) corresponding to $\B \cap \B'$.  The space $\phenome$ of phenomes is $\Power\bS$, also ordered by reverse inclusion.  The canonical projection $p: \U \ra \bS$ then lifts to a map $\pi : \system \ra \phenome$ sending $\B$ to $p(\B)$.

\begin{Pro}
 The map $\pi$ is a surjective veil.
\end{Pro}
\begin{proof}
  The map $\pi$ is clearly surjective and order-preserving.  The set $\{ B : S \supseteq \pi(B) \}$ has $p^{-1}(S)$ as a \emph{minimum} element with respect to $\supseteq$.
\end{proof}
Generative effects are typically sustained as:
\begin{equation*}
  \pi( \B \cap \C ) \neq \pi(\B) \cap \pi(\C).
\end{equation*}
Although the changes in $\C$ are not directly applied onto the subsystem, they do affect the subsystem through the other parts that are now concealed. The veil $\pi$ induces a map $i : \phenome \ra \system$ sending set $S$ to $S \times \bS'$.  The map $i\pi$ is then a kernel operator that destroys all the potential for effects to occur due to restrictions of the system in $\bS'$.

\subsubsection{Concealing parameters.}

For an additional interpretation of the example, let us suppose $\U = \bM\times \bL$.  The subsystem in play (whose universum corresponds to $\bM$) could represent manifest variables that are observable, and the rest (whose universum corresponds $\bL$) could represent latent variables or parameters that aid \emph{internally} in the workings of the system.  As such, changes in the internal parameters of a system affect the manifest variables.

\subsection{Concealing interdependence.}

The previous example may be further enhanced to understand interdependence between components of an interconnected system.  Let us again suppose that $\U = \bS \times \bS'$.  The decomposition yields two surjective maps $p : \U \ra \bS$ and $p' : \U \ra \bS'$.  The maps lift to separate veils $\pi : \system \ra \Power\bS$ and $\pi' : \system \ra \Power{\bS'}$ sending a behavior $\B$ to $p(\B)$ and $p'(\B')$, respectively.

We then define $\phenome$ to be $\Power\bS \times \Power{\bS'}$.  An element of $\phenome$ is thus a set $S \times S'$ with $S\subseteq \bS$ and $S' \subseteq \bS'$.  We finally define $\pi : \system \ra \phenome$ to be $p \times p'$ sending $\B$ to $p(\B) \times p(\B')$.

\begin{Pro}
  The map $\pi$ is a surjective veil.
\end{Pro}

\begin{proof}
   The map $\pi$ is clearly surjective and order-preserving. The set $\{B : S\times S' \supseteq  p(B)\times p'(B)\}$ has $p^{-1}(S) \cap {p'}^{-1} (S')$ (i.e., $S \times S'$) as a \emph{minimum} element with respect to $\supseteq$.
\end{proof}
And indeed, generative effects are typically sustained as:
\begin{equation*}
  \pi( \B \cap \B' ) \neq \pi(\B) \cap \pi(\B)
\end{equation*}
The veil $\pi$ induces a canonical inclusion map $i : \phenome \ra \system$. The map $i\pi$ is a kernel operator that destroys any potential interdependence between the components $\bS$ and $\bS'$.


\section{Forgetting characteristics.}

The second source of generative effects consists of forgetting characteristics, or properties, from the given system.  In such a setting, the space of phenomes tends to be larger than that of systems.  Indeed, phenomes then comprise the systems in concern as well as systems non-necessarily satisfying the desired characteristic to be forgotten. Generative effects emerge from the potential of the characteristic to enhance the interconnected system.

We can forget characteristics of a system, by defining a bigger set $\phenome$ containing $\system$. Every element of $\phenome$ can be seen as a potential system that is not forced to satisfy the forgotten characteristic.  Every element of $\phenome$ can then be treated as a partial observation of a system.  Such a partial observation can then be completed into a system satisfying the desired forgotten characteristic.  To this end, we require a map $c : \phenome \ra \phenome$ such that:
\begin{equation*}
  c(p) \in \system \text{ for all } p
\end{equation*}
and:
\begin{itemize}\setlength\itemsep{0em}
  \item[C.1.] $p \leq c(p)$ for all $p$. 
  \item[C.2.] If $p \leq p'$, then $c(p) \leq c(p')$, for all $p$ and $p'$.  
  \item[C.3.] $cc(p) = c(p)$ for all $p$. 
\end{itemize}
Notice the duality between the C.i in this section and the K.i in the previous ones.  First, the operator $c$ sends a partial observation to one that \emph{contains} it.  Second, the operator preserves the relation among the observation.  Third, the operator does not modify partial observations that are already systems. Again, an operator satisfying C.1, C.2 and C.3 is usually termed a closure operator. The property C.i coincides with property A.i in the contagion situation, whenever $\phenome$ is $\Power\Sigma$ for some set $\Sigma$.

The operator $c$ can intuitively be expected to sustain generative effects whenever $c (p \vee p') \neq c(p) \vee c(p')$ for some $p$ and $p'$.  The forgotten characteristic then indeed plays a role in the interaction of systems, to enhance the resulting combined system.

Given a closure operator $c$ on $\phenome$, the set $\system$ is identified with the set of fixed-points $\{p: c(p) = p\}$. We may then define a canonical inclusion $\iota : \system \ra \phenome$.  Every closure operator gives rise to a veil:

\begin{Pro} \label{Pro:ClosureVeil}
 The map $\iota$ is injective and order-preserving, and for every $p \in \phenome$, the set $\{s : p \leq \iota(s) \}$ has a minimum element.
\end{Pro}

\begin{proof}
 The map $\iota$ is injective by definition.  It is also clearly order-preserving.  Finally, the system $c(p)$ is the minimal element of $\{s : p \leq \iota(s) \}$.
\end{proof}
Conversely, every injective veil induces a closure operator on $P$:

\begin{Pro}\label{Pro:ClosureVeilConverse}
  If $\iota : \system \ra \phenome$ is an injective order-preserving map such that $\{s : p \leq \iota(s) \}$ has a minimum element for every $p$, then there exists a unique surjective map $q: \phenome \ra \system$ such that $q\iota$ is the identity on $\system$, and $\iota q$ is a closure operator on $\phenome$.
\end{Pro}

\begin{proof}
  For every $p \in \phenome$, define $q(p)$ to be the minimum element of $\{s : p \leq \iota(s) \}$.   The map $q\iota$ is the identity map as $q\iota(s)$ is the minimum element of $\{s' : \iota(s) \leq \iota(s')\}$, namely $s$ as $\iota$ is injective. The map $q$ is then surjective.  The map $\iota q$ is a closure operator as $\iota q(p)$ is the smallest element $\{\iota(s) : p \leq \iota(s)\}$.  The requirements C.i can be easily checked. Uniqueness of $q$ follows from Proposition \ref{Pro:Galois}(i).
\end{proof}
Finally, whether or not generative effects are sustained by the veil $\iota$ depends on the properties of the closure operator $c$.

\begin{Pro} If $s,s' \in \system$, then:
  \[\iota( s\vee_{\system} s') \neq \iota s \vee_{\phenome} \iota s'  \text{ iff } c (\iota s \vee_{\phenome} \iota s') \neq c (\iota s ) \vee_{\phenome} c(\iota s').\]
\end{Pro}

\begin{proof}
  We have $c(\iota s) = \iota s$ for every system $s$.  We also have $c(\iota s \vee_{\phenome} \iota s') = c\iota( s\vee_{\system} s')$ for every $s$ and $s'$.
\end{proof}
We next provide some example situations of forgetting characteristics.

\subsection{Zooming into a deductive system.}
We return to our contagion example.  Rather than considering the collection of all possible systems, and their interaction, we may zoom in on one particular system. Indeed, a system in the contagion example is itself a closure operator over $\Power\Sigma$, and thus may pave the way to generative effects.

Let $f: \Power\Sigma \ra \Power\Sigma$ be a map that satisfies A.1, A.2 and A.3.  We define $\phenome$ to be $\Power\Sigma$, the space of all configurations (of whether a node is infected or not) ordered by inclusion.  The space $\system$ will consist of all the configurations allowable by the dynamics, namely the fixed-points $\{ S : fS = S\} \subseteq \Power\Sigma$. As seen in Proposition \ref{pro:CompLattice}, those fixed-points form a complete lattice, where every pair of admissible configuration admits a join.  We then define $\iota$ to be the canonical inclusion $\system \ra \phenome$, and get:
\begin{Pro}
  The map $\iota$ is an injective veil.
\end{Pro}
\begin{proof}
  The statement follows from Proposition \ref{Pro:ClosureVeil}. Indeed, the map $\iota$ is injective by definition.  The system $f(S)$ is the minimal element of $\{s : S \leq \iota(s) \}$.
\end{proof}
Generative effects can sometimes be sustained, but not always. Whether or not the effects emerge depends on the properties of the closure operator $f$.

\subsubsection{No generative effects.}

Consider the following sequence of nodes.
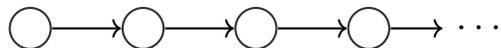
\begin{center}
\begin{tikzpicture}
  [scale=.5,auto=center,thick]
  \node[circle,draw=black!80!white,scale=1.4] (n1) at (0,0) {};
  \node[circle,draw=black!80!white,scale=1.4]  (n2) at (3,0)  {};
  \node[circle,draw=black!80!white,scale=1.4]  (n3) at (6,0)  {};
  \node[circle,draw=black!80!white,scale=1.4]  (n4) at (9,0)  {};
  \node[scale=1.2]  (n5) at (12,0)  {$\cdots$};
    \foreach \from/\to in {n1/n2,n2/n3,n3/n4,n4/n5}
    \draw[->] (\from) -- (\to);
\end{tikzpicture}
\end{center}
A node becomes infected if a node pointing to it becomes infected. Once a node is infected, it remains infected forever.
This system defines a map $f:\Power\Sigma \ra \Power\Sigma$ satisfying A.1, A.2 and A.3, and induces a veil $\iota$ as described earlier.
\begin{Pro}
 The veil does not sustain generative effects.
\end{Pro}
\begin{proof}
 We have $f(S \cup S') = f(S) \cup f(S')$ for every $S$ and $S'$ in $\Power\Sigma$.
\end{proof}
This system by itself does not exhibit any cascading phenomenon.  There might seem, however, to be an intuition for such a phenomenon only waiting to surface.  Nevertheless, when restricted to this particular system, the observable part \emph{prior to interaction} always determines the observable part \emph{after the interaction}.  The space of systems is linearly ordered in this case,  consisting of a maximal chain in $\Power\Sigma$. To allow the intuition to reappear, one needs to enlarge the space of systems.  The intuition of the phenomenon may be informally seen to come from the arcs of the directed graph. The arcs, however, are built into the situation, and cannot be modified.  If we enlarge our space of systems to include some systems that may not have arcs among nodes, then we can recover generative effects.  Enlarging so moves us one step closer to obtaining the whole class of maps satisfying A.1, A.2 and A.3 as a space of systems.

\subsection{Causality in systems.}

One may reconsider whether causality ought to be a concept of grandiose importance (see, e.g., \cite{RUSS1912}).  We view causality in this subsection as only intuitively expressing the notion of transitivity.  A situation of cause and effect will be considered as a transitive relation $\ra$ on a set $\Sigma$. The relation $a \ra b$ can be interpreted as ``$a$ causes $b$''.  Transitivity then abuts to: if $a \ra b$ and $b \ra c$ then $b \ra c$.  Cause-and-effect seems to be inherent in cascade-like phenomena.  They do appear, but only as a tangential special case of generative effects.  The intuition arises once we decide to forget the property of transitivity from a relation.

We consider a system to be a transitive relation on $\Sigma$.  The set $\system$ of transitive relations can be ordered by inclusion, and every pair of systems admits a least upper-bound in the poset. Two systems $R$ and $R'$ interact by taking the transitive closure $R \vee R'$ of their union $R \cup R'$. We can forget the transitivity property by \emph{embedding} $\system$ in a greater lattice $\phenome$ consisting of all relations on $\Sigma$. Two relations in $\phenome$ interact by yielding their union.  We define $\iota$ to be the canonical inclusion $\system \ra \phenome$.  Generative effects are sustained as, generally:
\begin{equation*}
  \iota(R \vee R') \neq \iota(R ) \cup \iota(R').
\end{equation*}
The transitivity property plays a role in the interaction, leading to more causal relations than what would typically be expected without it.  This generativity in the phenome is obtained by concealing characteristics in the systems.  Indeed, $\iota$ induces a closure operator on $\Sigma$ that sends a phenome to its transitive closure.

\subsubsection{Incorporating time.}

Time can be trivially incorporated by defining $\Sigma = E \times T$, where $E$ is a set representing events, and $T$ is an ordered set representing time.  One can further impose restrictions where $(e,t)$ may cause $(e',t')$ only if $t \leq t'$.  We will not dwell on developing such extensions in this paper. 

\subsection{Algebraic constructions.}

Closure operators abound in mathematics.  Those trivially include linear spans, convex hulls, and topological closure. As an example, let $G$ be an abelian group and suppose $S$ is its underlying set.  The space $\system$ will be the set $\sub(G)$ of subgroups of $G$ ordered by inclusion. The space $\phenome$ will be the set $2^S$ of subsets of $S$, again, ordered by inclusion.  Interaction in $\system$ is given by the linear span $+$ operator, while the interaction in $\phenome$ is given by set-union.  We can then define a closure operator on the set $S$, that sends subsets to the subgroup it generates.  The closure operator then defines an injective veil $\iota : \sub(G) \ra \Power S$.  Generative effects are sustained as, generally:
\begin{equation*}
   \iota( H + H') \neq \iota(H) \cup \iota(H')
\end{equation*}
The group axioms interact so as to produce more elements than what is only observed.  The veil is actually forgetting the group structure of the system.

\subsubsection{Forgetting might not create effects.}

Suppose $M$ is an $R$-module, and assume $G$ is its underlying abelian group.  The group $G$ is obtained by forgetting the multiplicative $R$-action of $M$. Let $\iota: \sub(M) \ra \sub(G)$ be the map that sends a submodule of $M$ to its underlying abelian group.  The map $\iota$ is an injective veil.  Generative effects are however never sustained. Indeed, for all submodules $N$ and $N'$, we have:
\begin{equation*}
  \iota(N + N') = \iota(N) + \iota(N').
\end{equation*}
The ring action plays no role that affects the underlying abelian group of a module.

\subsection{Universal grammar, languages and merge.}

It is argued in \cite{BER2015}---and more generally through the minimalist program, see e.g. \cite{CHO1993}, \cite{CHO1995} and related work---that the human ability of language universally arises from a single non-associative operation termed \emph{merge}.  For the illustrative purpose of this subsection, let us define a set $\Sigma$ of words.  Let $(\free\Sigma , \wedge)$ denote the free non-associative commutative algebra generated by the elements of $\Sigma$. We refer to an element of $\free\Sigma$ as a \emph{sentence}. Due to non-associativity, a \emph{sentence} is then not a linear concatenated string but rather a hierarchical object, a tree whose leaves are elements of $\Sigma$.  As expounded in \cite{BER2015}, language is fundamentally hierarchical and not associative, i.e., not concatenated (or not linear).  When externalized, say through speech, this sentence tends to be made concatenated.  A grammar is then seen as a subalgebra of $(\free\Sigma, \merge)$.  Such a subalgebra is thought to be generated by a set of \emph{sentences}.  The subalgebra can be ordered by inclusion to yield a join-semilattice $\big\<\sub \<\free\Sigma, \merge\>, \vee\big\>$.  The join $g \vee g'$ of two subalgebras $g$ and $g'$ is the subalgebra generated by $\{s \merge s' : s \in g \text{ and }s\in g'\}$.

These grammars possess an intuitive generative power, where the merge operator $\merge$ interacts with sentences to form new ones.  To capture it, we forget such a property.  Formally, we define a(n injective) veil, that sends a grammar (a subalgebra) to its underlying language (a set):
\begin{equation*}
   \iota : \big\<\sub \<\free\Sigma, \merge\>, \vee\big\> \ra \<\Power{\free\Sigma},\cup\>
\end{equation*}
Generative effects are sustained as, in general, we have:
\begin{equation*}
  \iota(g \vee g') \neq \iota(g) \cup \iota(g')
\end{equation*}
The discrepancy is caused by the effect of the merge operator. It is obtained by forgetting the characteristic that a grammar is equipped by such an operator, leaving only the underlying language.  It captures the \emph{generativity} of the grammars considered.

\section{On arbitrary veils.}

An arbitrary veil needs neither be surjective nor injective. Indeed, any combination of concealing mechanisms and forgetting properties can lead to an adequate veil. In general: 
\begin{Pro}
  The composition of two veils is a veil.
\end{Pro}
\begin{proof}
  The property V.1 is preserved under composition.  Let $\Phi_1 : S \ra Q$ and $\Phi_2: Q \ra P$ be veils.  Then, by V.2, for every $p \in P$, the set $\{ q : p \leq \Phi_2(q)\}$ has a minimum element $q_{min}$ and the set $\{s : q_{min} \leq \Phi_1(s)\}$ has a minimum element $s_{min}$.  If $p \leq \Phi_2\Phi_1(s)$, then $q_{min} \leq \Phi_1(s)$, and $s_{min} \leq s$. The set $\{s : p \leq \Phi_2\Phi_1(s)\}$ then has $s_{min}$ as a minimum element.
\end{proof}
A \emph{converse} also holds: every possible veil arises from a combination of concealing mechanisms and forgetting characteristics.

\begin{Pro}
  If $\Phi : \system \ra \phenome$ is a veil, then $\Phi$ admits a factorization $\Phi = \iota \pi$ such that $\pi : \system \ra Q$ is a surjective veil, and $\iota: Q \ra \phenome$ is an injective veil.
\end{Pro}

\begin{proof}  %
  Let $Q$ be the image set $\{\Phi(s) : s \in \system\}$. The map $\Phi$ factors as $\iota\pi$ through $Q$ with $\pi$ surjective, $\iota$ injective and both order-preserving.  As $\Phi$ is a veil, the set $\{s: p \leq \iota\pi(s)\}$ admits a minimum element for every $p$. The set $\{\pi(s): p \leq \iota\pi(s)\}$ then admits a minimum element as $\pi$ is order-preserving.  As $\pi$ is surjective, we get that $\{q \in Q: p \leq \iota q\}$ admits a minimum element for every $p$.  Similarly, for every $q \in Q$, the set $\{s: \iota(q) \leq \iota\pi(s)\}$ admits a minimum element.  As $Q \subseteq \phenome$, we have $\iota(q) \leq \iota(q')$ if, and only if $q \leq q'$. The set $\{s: q \leq \pi(s)\}$ then admits a minimum element.
\end{proof}
The map $\pi$ can be interpreted to conceal mechanisms in systems.  The map $\iota$ can then be interpreted as further forgetting characteristics from the partially concealed systems. Every situation of generative effects arises from a combination of the two cases. A veil then sustain generative effects whenever one of its factor components sustains those effects.

Finally, a good way to recognize a veil is by checking whether it preserves meets (i.e., greatest lower-bounds):

\begin{Pro} \label{Pro:veilasmeets}
 If $\system$ and $\phenome$ admit arbitrary meets, then: a map $\Phi : \system \ra \phenome$ is a veil if, and only if, it is order-preserving and preserves arbitrary meets.
\end{Pro}

\begin{proof}
  Suppose $\Phi$ is order-preserving and preserves arbitrary meets.  The set $\{s : p \leq \Phi(s)\}$ contains the maximum element of $\system$ (the meet of the empty set) and is thus non-empty.  The meet of all the elements in $\{s : p \leq \Phi(s)\}$ lies in this set and is its minimum element.  Conversely, suppose $\Phi$ is a veil. Let $S \subseteq \system$ be a subset.  We have $\Phi(\wedge S) \leq \wedge\Phi(S)$.  Also, $\{ s : \wedge \Phi(S) \leq \Phi(s) \}$ has a minimum element $s_{min}$. Obviously $s_{min} \leq s$ for every $s \in S$ and thus $s_{min} \leq \wedge S$. We then get $\wedge \Phi(S) \leq \Phi(\wedge S)$.
\end{proof}
In particular, $\Phi$ sends the maximum element in $\system$ to the maximum element of $\phenome$.  Indeed, the greatest lower-bound for the empty set yields the maximum element in $\phenome$.  Finally, if $\system$ and $\phenome$ only admit finite meets, the veil would preserve them. Indeed, the converse direction of the proof above goes through unchanged (by only considering $S$ to be finite) for the finite case.

\subsection{Properties and Galois connections.}

The property V.2 is crucial as it allows us to define a \emph{freest} (or simplest) reconstruction of a system from the phenome.  If $\Phi : \system \ra \phenome$ is a veil, then define $F : \phenome \ra \system$ (for free) such that:
\begin{equation*}
    F : p \mapsto \min\{s : p \leq \Phi(s)\}
\end{equation*}
If $\Phi$ is invertible, then $F$ would be the inverse of $\Phi$.  In the cases of interest, $\Phi$ is not invertible, and $F$ ought to be interpreted as the closest map that we could have as an inverse.  The map $F$ is said to be the \textbf{left} adjoint of $\Phi$, and the map $\Phi$ is said to be the \textbf{right} adjoint of $F$.  The pair $(F,\Phi)$ is termed a Galois connection.  We provide some properties related to Galois connections, and refer the reader to \cite{BIR1967} Ch. V, \cite{EVE1944}, \cite{ORE1944} and \cite{ERN1993} for a thorough treatment.  We will, however, not be explicitly using those properties.

\begin{Pro}\label{Pro:Galois}
  Let $\Phi$ be a veil, and $F$ be $p \mapsto \min\{s : p \leq \Phi(s)\}$, then:
  \begin{itemize}\setlength\itemsep{0em}
    \item[i.] For all $p$ and $s$, we have $F(p) \leq s \text{ iff } p \leq \Phi(s)$.
    \item[ii.] The map $F\Phi$ is a kernel operator on $\system$.
    \item[iii.] The map $\Phi F$ is a closure operator on $\phenome$.
    \item[iv.] The map $\Phi$ maps $s$ to the maximum of $\{s : s \leq F(p)\}$.
    \item[v.] For all $p$ and $p'$, we have $F(p \vee p') = F(p) \vee F(p')$.
  \end{itemize}
\end{Pro}

\begin{proof}
  We refer the reader to \cite{ERN1993} for proof of those statements, as well as other related statements.
\end{proof}
The map $F$ is the unique map such that the $(i.)$ holds. Furthermore, item $(ii.)$ recovers the kernel operator that conceals mechanisms in the system. Item $(iii.)$ recovers the closure operator that forgets characteristics of the systems.

One important consequence is:
\begin{Cor}
  If $\Phi: (\system, \leq) \ra (\phenome, \leq)$ is a veil, then its adjoint $F: (\phenome, \geq) \ra (\system, \geq)$ defines a veil on the dual preordered sets. \qed
\end{Cor}
Examples of Galois connections abound, especially when it comes to free construction of algebraic objects. Most are mathematical examples, but when interpreted appropriately yield us an intuition for cascade-like phenomena.  %

\section{Dynamical generative effects.}

We introduce, in this section, the notion of a dynamical veil.  Such veils can be used to incorporate temporal information in the phenome.   Aside from increasing modeling expressivity, dynamical veils can be used to \emph{spread} generative effects over (temporal) approximations of systems.  Such a spread may be used for a relative/successive analysis of generative effects.  This last direction will not be pursued in this paper.

A main argument in this paper is that generative effects enclose the intuitive notion of cascade effects.  The term \emph{cascade} however gives an impression of an evolving process.  The notion of time then seems to be an essential component for cascades.  However generative effects do not depend intrinsically on time.  Interconnection of systems does not depend on time either.  The goal of this interlude section is to aid in reconciling this view.  We thus introduce the notion of an ${\I}$-dynamical veil.  Whenever generative effects are sustained by such a veil, we may think of them as dynamically realized.

\begin{Def}
 Let $\I$ be a preordered set.  A veil $\Phi: \sS \ra \P$ is said to be an $\I$-dynamical veil if $\P$ is isomorphic to $P^{\I}$ (i.e., the preordered set of order-preserving maps $\I \ra P$) for some preordered set $P$.
\end{Def}
We often consider, in this section, sets $\I$ that are linearly ordered, i.e., where $\leq$ is antisymmetric for every $i,j \in \I$, either $i \leq j$ or $j \leq i$.  Such linear orders may be used to account for time, indexed by the elements of $\I$.

The notion of a system and the means of interconnecting systems remain unchanged.  The phenome is then obtained by \emph{reading} information from a system indexed by $\I$.  The space of systems then needs to be rich in (e.g., temporal) structure to support a \emph{meaningful} $\I$-dynamical veil.

\subsection{Revisiting contagion, dynamically.}

Let us reconsider the systems on contagion (or deduction) considered in Section \ref{sec:contagionSection} and studied in \cite{ADAM:AlgebraCascadeEffects}.  Let $\I$ be a preordered set, and let $P$ be a complete lattice.  We define $\L_{P^{\I}}$ to be the set of maps $f:P^{\I} \ra P^{\I}$ satisfying:
\begin{itemize} \setlength\itemsep{0em}
\item[A.1] If $a \in P$, then $a \leq f(a)$.
\item[A.2] If $a \leq b$, then $f(a) \leq f(b)$.
\item[A.3] If $a \in P$, then $f(f(a)) = f(a)$.
\end{itemize}
The set $\L_{P^{\I}}$ may be naturally ordered to form a lattice.  We may then define an order-preserving map:
\begin{equation*}
  \eval : \L_{P^{\I}} \ra P^{\I}
\end{equation*}
sending a system to its least fixed-point.   The map $\eval$ preserves arbitrary meets, and admits a left adjoint $\free : P^{\I} \ra \L_{P^{\I}}$.  The map $\eval$ may also be defined to act on $\L_P$, the maps $P\ra P$ satisfying A.1, A.2 and A.3, as done in Section \ref{sec:contagionSection} through $\Phi$.  

\subsubsection{Syntax and interpretation.}

We let $\Sigma$ be a finite set of $n$ elements.  The set $\Power\Sigma$ denotes the set of subsets (or powerset) or $\Sigma$.  We set $P$ to be $\Power\Sigma$, and consider $\I$ to be the (canonically) preordered set $\bZ^{\geq 0}$ of non-negative integers.

A system $f \in \L_{P^{\I}}$ may then be syntactically described by a map $\N : \Sigma \ra \Power{\Power\Sigma \times \bZ^{\geq 0}}$.  Indeed, every element $i$ of $\Sigma$ is attributed a collection of pairs $(S,d)$, where $S \subseteq \Sigma$ and $d \in \bZ^{\geq 0}$.  The interpretation is as follows.  Let $X_0, X_1, X_2, \cdots$ be subsets of $\Sigma$ where $X_0$ is the empty set.  If $S$ belongs to $X_m$, then $i \in X_{m+d}$.  The rule of course applies \emph{simultaneously} to all pairs $(S,d)$ for a given $i$, and for every element $i$ of $\Sigma$.

We may interpret $X_m$ to denote the elements that are active (infected or failed) at time $m$.  If the elements of $S$ are already active (infected or failed) at time $m$ (i.e., belong to $X_m$), then $i$ will become active (infected or failed) after $d$ time steps from $m$, i.e. at time $m+d$, belonging to $X_{m+d}$.

\subsection{From dynamical veils to veils.}

The phenome, in the case of an $\I$-dynamical veil, is thus a collection of related frames, or snapshots, taken from the system.  In case $\I$ is a linearly ordered set, the frames are successive snapshots of the system.  We can easily focus on one of the frames, forgetting others, and still recover a situation of generative effects.

Let $\I$ and $P$ be preordered sets.  We define $\pi_i : P^{\I} \ra P$ to be the canonical projection onto the $i$th component of $\I$.  Applying $\pi_i$ on the phenome, a collection of frames (or snapshot), amounts to only keeping the $i$th frame (or snapshot).

\begin{Pro}
  If $P$ is finitely complete (resp. finitely cocomplete) then for every $a, b \in P^\I$, we have $\pi_i(a \meet b) = \pi_i (a) \meet \pi_i (b)$ (\text{resp.} $\pi_i(a \join b) = \pi_i (a) \join \pi_i (b)$).
\end{Pro}

\begin{proof}
  Joins and meets in $P^\I$ are computed pointwise, if they exist in $P$.
\end{proof}
The projection $\pi_i$ is often also a veil.

\begin{Pro}
  If $P$ admits a minimum element, then $\pi_i$ is a veil.
\end{Pro}

\begin{proof}
  The map $\pi_i$ is clearly order-preserving. Let $0$ be the minimum element of $P$. For every $p$ in $P$, the map $q^*$ sending $j$ to $p$ if $i \leq j$ and to $0$ otherwise is the minimum of the set $\{q : p \leq \pi_i q\}$.
\end{proof}
This veil however does not sustain generative effects. This fact is a desirable feature.  Indeed, by composing a veil $\Phi: \sS \ra P^\I$ with $\pi_i$ we do not create additional generative effects.  We may then analyze $\Phi$ by separately analyzing the $\pi_i$'s.

There is also a more interesting means of recovering a single snapshot, achieved by aggregating everything.  In such a case, we retrieve the \emph{asymptotic} or \emph{limiting} behavior.  To this end, we define the map:
\begin{equation*}
  \colim : P^\I \ra P
\end{equation*}
that sends $a \in P^\I$ to $\vee_i a_i$.

The map $\colim$ is not always a veil.  As an example, suppose $\I =  \bZ^{\geq 0}$ (canonically preordered) and $P = \{0,1\}$ with $0 \leq 1$. Indeed, the set $\{a \in \{0,1\}^\bZ : 1 \leq \colim a\}$ does not admit a minimum element.  Regardless, the map $\colim$ is well behaved towards existing generative effects. Indeed:
\begin{Pro}
 The map $\colim$ preserves arbitrary joins, whenever they exist.
\end{Pro}

\begin{proof}
  Trivially $(\vee_i a_i) \vee (\vee_i b_i) = \vee_i (a_i \vee b_i)$ being the least upper-bound of $\{a_i\}\cup\{b_i\}$.
\end{proof}
We can thus aggregate phenomes, in the case of the dynamical contagion example over $\Power\Sigma$, without creating new generative effects. %
Let $\I$ be a preordered set and $P$ be a complete lattice.  Recall that $\L_{P^\I}$ (resp. $\L_P$) denotes the set of maps $P^\I \ra P^\I$ (resp. $P \ra P$) satisfying A.1, A.2 and A,3.

\begin{Pro}  For $a \in P$, let $a^\I$ denote the constant map in $\I \ra P$ with image $a$.  If $f \in \L_{P^\I}$, then the map $\agg f : P \ra P$:
  \begin{equation*}
    \agg f: a \mapsto \colim f(a^\I)
  \end{equation*}
  belongs to $\L_P$.
\end{Pro}
\begin{proof}
 (A.1) Since $a^\I \leq f(a^\I)$ by A.1 of $\L_{P^\I}$, then $\colim a^\I \leq \colim f(a^\I)$ as $\colim$ is order-preserving. (A.2) If $a \leq b$, then $a^\I \leq b^\I$.  It then follows that $f(a^\I) \leq f(b^\I)$ by A.2 of $\L_{P^\I}$, and thus $\colim f(a^\I) \leq \colim f(b^\I)$.  (A.3) Finally, $\colim f (\colim f(a^\I))^\I = \colim (\colim f(a^\I))^\I = \colim f(a^\I)$ as $f$ satisfies A.3 of $\L_{P^\I}$.
\end{proof}
In the case where $P = \Power\Sigma$ and $\I = \bZ^{\geq 0}$, the map $\agg$ can be interpreted to send a system with syntactic description $\N$ to one where all pairs $(S,d)$ for element $i$ are replaced by $(S,0)$.  In general, we recover the following commutative situation:

\begin{Pro}
  If $\I$ is a preordered set and $P$ is a complete lattice, then the diagram:
  \begin{equation*}
    \begin{CD}
      \L_{P^\I} @>\eval>> P^\I \\
      @V\agg VV               @VV\colim V\\
      \L_P  @>\eval>> P
    \end{CD}
  \end{equation*}
  commutes, i.e. $\colim \circ \eval = \eval \circ \agg$. 
\end{Pro}

\begin{proof}
  If $0$ denote the minimum element of $P$, then $\eval (\agg f) = (\agg f) (0) = \colim f(0^\I) = \eval f$.
\end{proof}
We have been \emph{projecting} the phenome in $P^\I$ to a phenome in $P$, while making sure that joins are preserved.  Join-preservation then does not create new generative effects.  The projection however will often remove some of the original generative effects.  One ought to then think of such procedures as focusing on a particular refined aspect of the dynamical phenome.

\subsection{From veils to dynamical veils.}

Any veil is trivially a dynamical veil over the one-point poset.  We can however obtain a veil that is non-trivially dynamical by considering the systems to be given by a \emph{filtration}.

\begin{Def}
 Let $\I$ be a preordered set.  An $I$-filtration of a system $s$ in a preordered set $\sS$ is an order-preserving map $F : \I \ra \sS$ such that $\colim F = s$.
\end{Def}
An $\I$-filtration then provides a \emph{successive} approximation for the system $s$.  Every veil $\Phi : \sS \ra \P$ induces a canonical veil $\Phi^\I : \sS^\I \ra \P^\I$.  Defining a system by an $\I$-filtrations can be seen to equip it with adequate temporal information.

\subsubsection{In the behavioral approach.}

Through the lens of the behavioral approach to systems theory, such temporal information can be seen as a further refinement of constraints. Let $\I$ be a linearly ordered set of $n$ elements, and consider the set $\bU = U_1 \times \cdots \times U_n$.  For every $i$, let $p_i : \bU \ra U_i$ be the canonical projection. For every Willems' system $(\bU,\B)$, the sets:
\begin{equation*}
  F_i \B = \pi_i^{-1} \pi_i \B
\end{equation*}
then define an $\I$-filtration of $\B$.  Every set $U_i$ can be seen to represent a variable, an $\I$-filtration can then be seen to successively \emph{grow} the constraints to connect different variables.

\subsection{Why care about dynamical veils?}

Dynamical veils can be used to incorporate \emph{time} in generative effects.  There are, however, various other ways of incorporating time in generative effects, such as having the phenome contain timed trajectories.  However, going from a veil to a dynamical veil, by resolving a system into an $\I$-filtration, may allow us to \emph{spread} generative effects.  The eventual goal fully developed in \cite{ADAM:Dissertation} Ch 8 and exemplified in \cite{ADAM:Dissertation} Ch 2, 6 and 7 is to develop \emph{cohomology} theories for understanding such effects.  These dynamical veils may allow us to develop \emph{relative} theories.  It other terms, it may allows to ask (and answer) the following informal question:  suppose we have observed a cascade, and its effects, up to time $T$, what new effects resulting from the cascade will appear at time $T+m$?  This direction will, however, not be further pursued in the paper.

\section{Factorization and lifts.} 
The second condition V.2 of a veil may seem to be restrictive.  Some situations may be formalized in a way that does not yield such a condition.  We show how one can recover a veil from non-veil-like situations.

\subsection{Factoring and retrieving the intuition.}

Let $P$ and $Q$ be cocomplete preordered sets (admitting arbitrary joins) and let $f : P \ra Q$ be an order-preserving map.  It can be the case that $P$ and $Q$ contain elements that make property V.2 fail for $f$, but that are irrelevant to any situation of generative effects possibly suggested by $f$.  There exists a systematic way to get rid of such elements, and potentially retrieve a hidden veil.

\subsubsection{On the $P$ side.} Define a binary relation $\sim$ on $P$ such that $p \sim q$ if, and only if, $f(p \vee x) = f(q \vee x)$ for all $x \in P$.

\begin{Pro}
  The relation $\sim$ is an equivalence relation.
\end{Pro}
\begin{proof}
 The relation can be trivially checked to be reflexive, symmetric and transitive.
\end{proof}
The relation $\sim$ is further compatible with the structure of $P$ when viewed as a join semilattice.

\begin{Pro}
  The relation $\sim$ is a congruence relation on $(P,\vee)$, i.e., if $a \sim b$ and $a' \sim b'$ then $a \vee a' \sim b\vee b'$.
\end{Pro}
\begin{proof}
  Suppose $a \sim b$ and $a' \sim b'$. For every $x \in P$, we  have:
  \begin{equation*}
    f(a \vee a' \vee x) = f(b \vee (a' \vee x)) = f( a' \vee (b \vee x) ) = f( b' \vee b \vee x).
  \end{equation*}
  The equalities follow from commutativity and associativity of $\vee$.
\end{proof}
To get a better understanding of $\sim$ we note:

\begin{Pro}
 If $a \leq c$ and $a \sim c$, then: for every $b$ where $a \leq b \leq c$, we have $b \sim c$.
\end{Pro}

\begin{proof}
  If $a \leq b \leq c$, then for every $x$, we have $f(a \vee x) \leq f(b \vee x) \leq f( c \vee x )$.
\end{proof}

\begin{Pro}
  If $a \sim b$, then $b \sim a\vee b$.
\end{Pro}

\begin{proof}
  Indeed, we have that $a \sim b$ and $b \sim b$. The result then follows by congruence. Or directly, we have $f(a \vee b \vee x) = f(b \vee b \vee x) = f(b \vee x)$ 
\end{proof}
In particular, a maximal element (if it exists) of a congruence class of $\sim$ is the (unique) maximum element of the class.  The congruence classes induced by $\sim$ define a partition of $P$. These congruence classes inherit an order to yield a join-semilattice $P_\sim$.  The order relation $\leq$ on $P_\sim$ is defined as $C \leq C'$ if, and only if, there is an $a \in C$ and $a \in C'$ such that $a \leq a'$.  Equivalently, the join operation is defined as: if $a$ and $a'$ are in the classes $C$ and $C'$, then $C \vee C'$ is the congruence class containing $a \vee a'$.  The join operation is well defined as $\sim$ is a congruence relation. Let $\pi : P \ra P_\sim$ be the order-preserving surjective map that sends an element of $P$ to its congruence class. 
\begin{Pro}
  The map $\pi$ commutes with finite joins.
\end{Pro}

\begin{proof}
  The statement follows by construction of $P_\sim$, namely from the fact that $\sim$ is a congruence relation. 
\end{proof}

\subsubsection{On the $Q$ side.} Define $\hat{Q}$ to be sub-joinsemilattice of $Q$ generated by $f(P)$.  Namely we define $\hat{Q}$ to be the elements of $f(P)$ along with all possible finite joins ordered by the partial order of $Q$.  Let $\iota: \hat{Q} \ra Q$ be the canonical order-preserving injective map.

\begin{Pro}
 The map $\iota$ commutes with finite joins.
\end{Pro}

\begin{proof}
  The statement immediately follows from the construction of $\hat{Q}$.
\end{proof}

\subsubsection{Combined.}  The maps $\pi$ and $\iota$ may be used to factorize $f$.

\begin{Pro}
  There exists a unique map $g: P_{\sim} \ra \hat{Q}$ such that the diagram:
  \begin{equation*}
    \begin{CD}
      P_\sim @>g>> \hat{Q}\\
      @A\pi AA       @VV\iota V\\
      P @>f>> Q
    \end{CD}
  \end{equation*}
  commutes, i.e., $f = \pi g \iota$.
\end{Pro}

\begin{proof}
  For every class $C$ in $P_\sim$, let $\alpha C$ denote a fixed element of $C$. Define $g$ to be $C \mapsto f(\alpha C)$.  The diagram commutes, and the map $g$ is unique as $\iota$ is injective and $\pi$ is surjective.
\end{proof}
We then get:

\begin{Cor} For all $a,b \in P$, we have:
  \begin{equation*}
    f (a \vee b) \neq f( a) \vee f(b) \quad \text{ iff } \quad  g (\pi a \vee_{P_\sim} \pi b)  \neq  g (\pi a) \vee_{\hat{Q}} g(\pi b).
  \end{equation*}
\end{Cor}
Furthermore:
\begin{Pro}
  If $f : P \ra Q$ is injective, then (i) $\pi$ is the identity and, (ii)  $g$ is a veil if, and only if, for every $p$ and $p'$, we have:
  \begin{equation*}
    \text{if } f(p) \leq f(p'), \text{then } p \leq p'.
  \end{equation*}
\end{Pro}

\begin{proof}  
  Let $f$ be injective.  Then $p \sim p'$ implies $p = p'$, and $\pi$ is thus the identity.

  Suppose that $g$ is a veil. %
If $f(p) \leq f(p')$, then $g(p) \leq g(p')$. The greatest lower-bound of $g(p)$ and $g(p')$ then exists and is $g(p) \wedge g(p') = g(p)$.  If we consider the set $T = \{t : g(p) \wedge g(p') \leq g(t)\}$, then it follows that $p, p' \in T$ and $p$ is the (unique) minimum of $T$.  We then have $p \leq p'$.

  Conversely, suppose that $f(p) \leq f(p')$ implies $p \leq p'$ and consider the set $\{p : q \leq g(p)\}$ for $q\in \hat{Q}$.   If $q \in f(P)$ then, as $f$ is injective, the set $\{p : q \leq g(p)\}$ admits a unique minimum, the preimage of $q$ with respect to $f$.  If $q \notin f(P)$, then $q = \vee_i f(p_i)$ is a finite join of elements in $f(P)$.  We also have: \begin{align*}
    \{ p : \vee_i f(p_i) \leq g(p)\} &= \cap_i \{ p : f(p_i) \leq g(p)\} \\
    &= \cap_i \{ p : f(p_i) \leq f(p)\} \\
    &= \cap_i \{ p : p_i \leq p\}\\
    &= \{p : \vee_i p_i \leq p\}
  \end{align*}
This set has $\vee_i p_i$ as a minimum element.
\end{proof}
Suppose that $P$ is finite.  Then, every congruence class $C$ in $\sim$ admits a maximum element, the join of all its elements.  We define $c : P \ra P$ to be the map that sends an element $p$ to the maximum element of its congruence class $\pi(p)$.  Notice that $c$ is a closure operator on $P$.

\begin{Pro}
  Let $P$ be finite. If $f : P \ra Q$ is surjective, then (i) $\iota$ is the identity, and (ii) $g$ is a veil if, and only if, $f( \wedge_i p_i) = \wedge_i f(p_i)$ for every (finite) collection $\{p_i\} \subseteq c(P)$.
\end{Pro}

\begin{proof}
  If $f$ is surjective, then $\hat{Q} = f(P) = Q$.  The rest follows from Proposition \ref{Pro:veilasmeets}. Indeed, $c(P)$ is isomorphic to $P_\sim$ and $g$ is the restriction of $f$ to $c(P)$.
\end{proof}
The finiteness condition can be alleviated through adequate technical care.

\subsection{All order-preserving maps can be lifted to veils.}

A factorization, as done in the previous section, need not always yield a veil.  We will often have a map that is not necessarily a veil, but would still like to interpret the situation as one exhibiting generative effects.  If we have a map that does not satisfy V.2, then some phenome will not have a minimum system that explains it.  It will have multiple minimal systems explaining it. However if we can treat the multitude of systems as one \emph{ambiguous} system, we can recover uniqueness.

In this section, we show that we can always lift an arbitrary order-preserving map to a veil between a lifted space of systems and a lifted space of phenomes. The relevant properties, namely whether or not it sustains generative effects, are preserved in the lift.  We do lose something by this \emph{completion}, as now the lifted space contains objects that we cannot necessarily interpret as systems.  That need not be a nuisance as interaction of the interpretable systems is preserved.  The generality of the lift can however restrict our ability to find tight structures for the situation.

\begin{Def}
 A filter (or upper set) $J$ of a preordered set $P$ is a subset of $P$ such that: if $p \leq p'$ and $p \in J$, then $p' \in J$.
\end{Def}
In particular:

\begin{Pro}
  If $P$ is finitely cocomplete, then: $J$ is a filter of  $P$ if, and only if, $p \vee J \subseteq J$ for all $p\in P$.
\end{Pro}

\begin{proof}
  If $J$ is a filter of $P$, then for every $j \in J$, we have $j \leq p \vee j$ and thus $p \vee j \in J$.  Conversely, if $p \leq p'$ and $p \in J$, then $p' = p' \vee p \in p' \vee J \subseteq J$.
\end{proof}
We denote by $\J(P)$ the preordered set of filters of $P$. The set $(\J(P), \supseteq)$ ordered by reverse inclusion  is necessarily a lattice, that admits arbitrary, joins (through set-intersection $\cap$) and meets (through set-union $\cup$).  In particular, $\J(P)$ is a distributive lattice as $\cap$ and $\cup$ distribute over one another.

If $p \in P$, we define $\<p\>$ to be the filter generated by $p$. Namely: $\<p\> = \{ p\vee a : a \in P\}$. The element $p\in P$ is then represented by $\<p\>$ in $\J(P)$. Indeed:

\begin{Pro}\label{pro:commutesJoin}
  If $P$ is finitely cocomplete, then: for all $p, p' \in P$, we have $\<p \vee p'\> = \<p \> \cap \<p'\>$.
\end{Pro}

\begin{proof}
   If $a \in \<p \> \cap \<p'\>$, then $p \vee p' \leq a$, and so $a \in \<p \vee p'\>$.  Conversely, we trivially have $\<p \vee p'\> \subseteq \<p \> \cap \<p'\>$.
\end{proof}
Note that lifting to $\J(P)$ does not preserve meets. Meets however are not essential throughout the theory, they just happen to be a convenience.

Let $f : P \ra Q$ be an order-preserving map. If $I$ is a filter, then $f(I)$ does not have to be a filter. We define $\J(f): \J(P) \ra \J(Q)$ to be the map that sends a filter $I$ to the filter closure $\<f(I)\>$ of $f(I)$.

\begin{Pro}\label{Pro:Jfp}
  We have $\J(f)\<p\> = \< f(p)\>$ for all $p$.
\end{Pro}

\begin{proof}
 The set $f\<p\>$ contains $f(p)$ as a minimum element.
\end{proof}
Furthermore:

\begin{Pro}
 The map $\J(f): \<\J(P),\supseteq\> \ra \<\J(Q),\supseteq\>$ is a veil.
\end{Pro}

\begin{proof}
  The map $\J(f)$ is clearly order-preserving.  Furthermore, the set $\{I : J \supseteq \J(f)(I) \}$ has $\big\<\{p : J \supseteq \<fp\>\}\big\>$ as a \emph{minimum} element with respect to $\supseteq$. 
\end{proof}
Finally, the potential for generative effects is preserved by the lift.

\begin{Pro}
  If $P$ and $Q$ are finitely cocomplete, then:
  \begin{equation*}
    f(p \vee p') \neq f(p) \vee f(p') \quad \text{iff} \quad \J(f)(\<p\> \cap \<p'\>) \neq \J(f)(\<p\>) \cap \J(f)(\<p'\>).
  \end{equation*}
\end{Pro}

\begin{proof}
  We have $f(p \vee p') \neq f(p) \vee f(p')$ if, and only if, $\<f(p \vee p')\> \neq \<f(p) \vee f(p')\>$. By Proposition \ref{pro:commutesJoin}, we get $\<f(p) \vee f(p')\> = \<f(p)\> \cap \<f(p')\>$ and $\<p \vee p'\> = \<p \> \cap \<p'\>$. The rest then follows by Proposition \ref{Pro:Jfp}.
\end{proof}
The $\J$ operator however disregards any information on whether or not $f$ already satisfies V.2.  For instance, if $f$ already satisfies V.2 we should expect that we ought not need a lot of elements to add in the lifted space while completing the space, if anything at all. Such  \emph{controlled} lifts can be achieved via the use of Grothendieck topologies on posets.  Such a direction will not be pursued in this paper.  It is however an important direction to develop: it provides a solid path towards toposes and the theory of sheaves (see e.g. \cite{SGA4}).

\subsubsection{Remark.} Notice that if the spaces $P$ and $Q$ were not cocomplete, the lift will \emph{create} joins in the lattice of filters.

\section{Concluding remarks.}
The development of generative effects has been carried out in the restrictive case of preordered sets.  We can achieve greater generality, and higher expressiveness, by having our spaces of systems and phenomes be categories.  We refer the reader to \cite{ADAM:Dissertation} Ch 8 for the details.

In the general level, the presence of generative effects can be linked to a loss of exactness.  This loss can be recovered via methods in homological algebra. We can extract algebraic objects that encode a system's potential to produce generative effects.  We can then use those objects to characterize the phenomenon, and link the behavior of the interconnected system to that of its separate systems.  This direction is carried out in \cite{ADAM:Dissertation}.

  \bibliographystyle{alpha}
  \bibliography{MyBiblio}

\end{document}